\documentclass[runningheads]{llncs}
\usepackage[T1]{fontenc}

\usepackage{amsmath}
\usepackage{amssymb}
\usepackage{tikz}
\usepackage{tikz-cd}
\usepackage{color}   
\usepackage{enumerate}
\usepackage{ifthen}
\usepackage{bm}
\usepackage{import}
\usepackage{paralist}
\usepackage{bbding}

\usetikzlibrary{positioning}
\usetikzlibrary{shapes.geometric, arrows}
\usetikzlibrary{decorations.pathreplacing}

\newtheorem{observation}{Observation}

\newcommand{\N}{\mathbb{N}}

\newcommand{\set}[1]{ \left\{ #1 \right\} }
\newcommand{\binset}{\set{0,1}}

\newcommand{\Walk}[2]{\mathrm{walk}_{#1,#2}}
\newcommand{\CycleWalk}[2]{\mathrm{cwalk}_{#1,#2}}

\newcommand{\PCF}[3]{\mathrm{PC}^{#1}_{#2,#3}}
\newcommand{\Cycle}[1]{\mathrm{Cycle}_{#1}}

\newcommand{\Index}{\mathrm{Index}}
\newcommand{\ORIndex}{\mathrm{OR\text{-}Index}}
\newcommand{\disj}{v_\mathrm{\tt disj}}
\newcommand{\intersect}{v_\mathrm{\tt int}}

\begin{document}

\title{Pointer Chasing with Unlimited Interaction}
\author{Orr Fischer\inst{1}\orcidID{0009-0007-4197-015X} \thanks{Supported in part by the Israel Science Foundation, Grants No. 1042/22 and 800/22).}\and
Rotem Oshman\inst{2}\orcidID{0009-0007-5065-5557}\thanks{Supported in part by the Israel
Science Foundation, Grant No. 2801/20.}
 \and
Adi Ros{\'{e}}n\inst{3}\thanks{Most of the work by this author was done while with FILOFOCS, CNRS, Israel.}
\and
Tal Roth\inst{4}\orcidID{0009-0000-7094-8348} \thanks{Supported in part by the Israel
Science Foundation, Grant No. 2801/20.}\Envelope}
\authorrunning{O. Fischer et al.}
\institute{Bar Ilan University, Israel \\ \email{fischeo@biu.ac.il} 
\and
Tel Aviv university, Israel \\ \email{roshman@tau.ac.il}
\and
IRIF, CNRS and Université Paris Cité, France \\
\email{adiro@irif.fr} \and
Tel Aviv university, Israel \\ \email{roth1@mail.tau.ac.il}
}
\maketitle

\begin{abstract}

Pointer-chasing is a central problem in two-party communication complexity:
given input size $n$ and a parameter $k$,
the two players Alice and Bob are given functions $N_A, N_B: [n] \rightarrow [n]$, respectively,
and their goal is to compute the value of $p_k$, where $p_0 = 1$, $p_1 = N_A(p_0)$, $p_2 = N_B(p_1) = N_B(N_A(p_0))$, $p_3 = N_A(p_2) = N_A(N_B(N_A(p_0)))$
and so on, applying $N_A$ in even steps and $N_B$ in odd steps, for a total of $k$ steps.
In some versions of the problem, the final output is not $p_k$ itself, but rather some fixed function $f(p_k)$ of $p_k$.
It is trivial to solve the problem using $k$ communication rounds, with Alice speaking first, by simply ``chasing the function'' for $k$ steps.
Many works have studied the communication complexity of pointer chasing, although the focus has always been on protocols with $k-1$ communication rounds, or with $k$ rounds where Bob (the ``wrong player'') speaks first.
Many works have studied this setting giving sometimes tight or near-tight results.

In this paper we study the communication complexity of the pointer chasing problem when the interaction between the two players is unlimited, i.e., without any restriction on the number of rounds. Perhaps surprisingly, this question was not studied before, to the best of our knowledge.
Our main result is that the trivial $k$-round protocol is nearly tight (even) when the number of rounds is not restricted:
we give a lower bound of $\Omega(k \log (n/k))$ on the randomized communication complexity of the pointer chasing problem with unlimited interaction, and a somewhat stronger lower bound of 
$\Omega(k \log \log{k})$ for protocols with zero error.

 When combined with prior work, our results also give a nearly-tight bound on the communication complexity of protocols using at most $k-1$ rounds, across all regimes of $k$;
 for $k > \sqrt{n}$ there was previously a significant gap between the upper and lower bound.

\end{abstract}

\section{Introduction}
\label{sec:intro}
\emph{Pointer chasing} is a natural and well-studied problem in communication complexity, which  was used to demonstrate the inherent sequential nature of certain distributed tasks where the input is partitioned between a number of players.
In particular, it was used to
demonstrate 
that certain tasks must be solved  ``step by step'' and cannot be parallelized;
such tasks require a large number of back-and-forth communication rounds,
\footnote{See Section~\ref{sec:prelim} for a more formal definition of this and other notions in communication complexity.}
and therefore
limiting the number of rounds of communication between two parties may have a dramatic effect on  communication complexity.
Pointer chasing has found many applications, including a proof for the monotone constant-depth hierarchy for Boolean circuits~\cite{NisanW91}, 
and lower bounds in distributed computation~\cite{NanongkaiSP11}, streaming algorithms~\cite{FeigenbaumKMSZ08,GuruswamiO13,AssadiCK19}, Matroid intersections~\cite{Harvey08}, 
data structures~\cite{MiltersenNSW98,SenV08},
and more.

It is easiest to informally describe the pointer chasing problem using the terminology of graphs. For  integer numbers $n$ and $k \leq n$,
in the two-party $k$-step pointer chasing problem,
we have a directed bipartite graph
$G=(U,W,E)$, where $|U|=|W|=n$, $U \cap W = \emptyset$,
 and the out-degree of all nodes is exactly $1$. 
Starting from any node $u \in V$, there is a natural notion of a \emph{$k$-walk} starting at that node,
where we follow the outgoing edges for $k$ steps,
alternating between nodes in $U$ and in $W$.
In the  $k$-step pointer chasing problem,
Alice receives all the edges emanating  from nodes in $U$ and Bob receives all edges emanating from nodes in $W$, and the task is to
find the endpoint $v$ of a $k$-step walk starting from some fixed node $u \in U$ on Alice's side of the graph
 (or sometimes, to compute a function $f(v)$ of that node, for some function $f$ that is fixed in advance and known to both players). 
The trivial protocol for $k$-step pointer chasing requires $k$ rounds, with Alice speaking first: 
the players follow the $k$-walk by having each player in their turn announce the vertex to which the walk moves, with Alice speaking in odd turns (since she knows the edges going from left to right) and Bob speaking in even turns (since he knows the edges going from right to left).
The total communication is $O(k \log n)$ bits.
We remark that when $k = \Theta(n)$, the players can simply send each other their full inputs, and this requires
one round of communication and $O(n \log n)$ bits.

The pointer chasing problem was first introduced by Papadimitriou and Sipser~\cite{PapadimitriouS84}, who showed that  for the case of $k=2$, any one-round protocol requires exponentially more communication than the trivial two-round protocol.
The communication complexity of ($k$-step) pointer chasing for $(k-1)$-round protocols, or for $k$-round protocols where Bob speaks first, was subsequently studied in a number of works 
\cite{NisanW91,DammJS96,PonzioRV99,Klauck00,Yehudayoff20,MaoYZ24}. 
At present,
the state of the art for   {\em deterministic} communication protocols
is an upper bound of $O \left(n \log^{(k-1)}n + k \log n \right)$~\cite{DammJS96}
and a lower bound of $\Omega(n - k \log n)$~\cite{NisanW91},
which holds for any $k$.
For  {\em randomized} communication protocols,
the best known upper bound is
$O\left( \left(\frac{n}{k} + k \right) \log n \right)$~\cite{NisanW91}, and
the best lower bound is $\Omega(n/k+k)$~\cite{MaoYZ24}.
Crucially, these lower bounds apply to protocols with \emph{exactly} $k-1$ rounds where Alice speaks first,
or to protocols with \emph{exactly} $k$ rounds where Bob speaks first.
See Section~\ref{sec:related} for more details on these and other results.

Perhaps surprisingly, the communication complexity of the pointer chasing problem with {\em unlimited interaction}, i.e., without any upper or lower bound on the number of rounds, has, to the best of our knowledge, not been studied to date. 
In this work we give lower bounds on the communication complexity of the $k$-step pointer chasing problem when there is {\em no restriction on the number of rounds}. 
We give the following two main results:
\begin{itemize}
\item An $\Omega(k \log (n/k))$ lower bound on the expected communication of randomized protocols that have a constant but non-zero error probability, and
\item An $\Omega(k \log \log{k})$ lower bound on the expected communication of  zero-error randomized protocols.
\end{itemize}

As a by-product of our results, we also get new, essentially tight, results for the communication complexity of the  pointer chasing problem when the protocol is restricted to {\em at most} $k-1$ rounds and Alice speaks first, or {\em at most} $k$ rounds and Bob speaks first, for certain regimes of $k$ and $n$. Again, this is in contrast to all previous lower bounds results which were proved for the case when the protocol uses {\em exactly} this number of rounds. While this difference seems at first sight to be of no great consequence, some of the lower bounds mentioned above, such as those of~\cite{Klauck00} and~\cite{MaoYZ24}, do rely on this fact, and their proofs do not hold for the more general case.
Our results imply new, nearly-tight lower bounds for the regime
where 
 $k > \sqrt{n}$, closing
 a significant gap (when $k \gg \sqrt{n}$)
 between the upper bound of $O\left( \left(\frac{n}{k} + k \right) \log n \right)$~\cite{NisanW91} and the lower bound of $\Omega(n/k)$~\cite{MaoYZ24} (restated for the case of {\em at most} $k-1$ rounds; see Section~\ref{sec:related} for more details).

 Another by-product of our results is a negative answer to a question posed very recently~\cite{MaoYZ24}:
 it is conjectured in~\cite{MaoYZ24} that the upper bound of  $O\left( \left(\frac{n}{k} + k \right) \log n \right)$ from~\cite{NisanW91} is not tight for $k=\omega(\log n)$,
 and the $\log n$ factor can be removed (this upper bound is indeed not tight for $k=o(\log n)$).
 The lower bound that we prove in Theorem~\ref{thm:LB_sublinear_k} below shows that this is not the case for $k = \Theta\left( n^\delta \right)$, where $1/2 < \delta < 1$ is a constant: the factor of $\log n$ is inherent.

 Our results are proved using two simple reductions: for non-zero error protocols
 we reduce from the $\ORIndex$ problem,
 and for zero-error protocols
 we reduce from $\mathrm{Cycle}$ (we review these problems in Section~\ref{sec:prelim}).

\subsection{Related Work}
\label{sec:related}

The known upper and lower bounds on $k$-step pointer chasing are summarized in Tables~\ref{table:known_lower_bounds} and~\ref{table:known_upper_bounds} below.
In the tables we list upper and lower bounds for {deterministic protocols with worst-case error (listed as ``deterministic'' in the tables),
for {randomized} protocols,
and for {deterministic} protocols with \emph{distributional error} (listed as ``distributional'' in the tables),
where the protocol only needs to succeed with high probability over  inputs  drawn from some fixed distribution (in this case, typically the uniform distribution).
Distributional lower bounds are the strongest, as they imply both deterministic worst-case lower bounds and, by Yao's principle, randomized lower bounds with worst-case error.

We emphasize that the lower bounds listed in Table~\ref{table:known_lower_bounds} all apply only to protocol with \emph{exactly} the number of rounds listed in the table.
This may be confusing at first, as one might expect that a lower bound that applies to protocol with \emph{exactly} $R$ communication rounds would also apply to protocols with \emph{at most $R$} communication rounds, but in fact this is not necessarily the case.
The distinction is most significant when it comes to the lower bounds of~\cite{Klauck00} and~\cite{MaoYZ24}, where the proofs make explicit use of the fact that the protocol has exactly some number of rounds.
It is possible to adapt these proofs so that they apply to any protocol with \emph{at most $k$} communication rounds,
but this comes at the cost of $k$ bits in the lower bound, which becomes $\Omega(n/k)$ in both cases.
Therefore, when $k \gg \sqrt{n}$, the lower bounds of~\cite{Klauck00} and of~\cite{MaoYZ24} are significantly weaker than the lower bounds of $\tilde{\Omega}(k)$ that we prove in the present paper.

From a technical perspective, one reason that prior work has hit an obstacle at $k = \sqrt{n}$ is that most of it has worked with   the \emph{uniform input distribution}, where the outgoing edge of each node in the bipartite graph goes to a uniformly random node on the other side.
It is not hard to see that under that distribution within the first $k \approx \sqrt{n}$ steps the walk enters a \emph{cycle},
which means that the players do not need to ``keep walking'' and may instead reuse the information that they have already learned, without further communication.
Thus, proving lower bounds that grow with $k$ when $k \geq \sqrt{n}$ requires a new idea.

\begin{table}[h!]
\centering
\caption{Related work: lower bounds}
\label{table:known_lower_bounds}
\begin{tabular}{|c|c|c|} 
 \hline
 Paper & lower bound & Comments \\  
 \hline
\cite{PapadimitriouS84}  & $\Omega(n)$ & $k=2$, deterministic one-way\\ 
 \hline
\cite{DurisGS84}  & $\Omega(n/k^2)$ & deterministic, $k-1$ rounds\\ 
 \hline
\cite{NisanW91}  & $\Omega(n - k\log n)$ & deterministic, $k$ rounds (Bob speaks first)\\
                 & $\Omega((n/k^2) - k\log n)$ & randomized, $k$ rounds (Bob speaks first) \\
 \hline
\cite{PonzioRV99} & $\Omega(n \log^{(k-1)}n)$ & constant $k$, deterministic, $k$-rounds (Bob speaks first)\\ 
 \hline
\cite{Klauck00} &  $\Omega((n/k) + k)$ & randomized, $k$ rounds (Bob speaks first) \\ 
 \hline
\cite{Yehudayoff20} & $\Omega((n/k) - k\log n)$ & distributional, $k$ rounds (Bob speaks first) \\ 
 \hline
\cite{MaoYZ24} &  $\Omega((n/k) +k)$ & distributional, $k-1$ rounds 
\\ 
 \hline
\end{tabular}
\end{table}

\begin{table}[h!]
\centering
\caption{Related work: upper bounds}
\label{table:known_upper_bounds}
\begin{tabular}{|c|c|c|} 
 \hline
 Paper & Upper bound & Comments \\  
 \hline
 Trivial protocol & $O(k \log n)$ & $k$ rounds (Alice speaks first), deterministic \\ 
 \hline
 \cite{DammJS96} & $O \left(n \log^{(k-1)}n + k \log n \right)$ & $k$ rounds (Bob speaks first),  deterministic \\
 \hline
 \cite{NisanW91} & $O\left( \left(\frac{n}{k} + k \right) \log n \right)$ & $k-1$ rounds, randomized \\
 \hline
\end{tabular}
\end{table}

\subsection{Other Models}
In addition to the standard two-party communication model,
pointer chasing was also studied in the quantum communication model \cite{Klauck00,KlauckNTZ01,JainRS02,JainRS02b} and in the multiparty number-on-the-forehead model \cite{NisanW91,DammJS96,Gronemeier06,ViolaW09,Chakrabarti07,BrodyC08,Brody09,Jastrzebski14,Liang14,BrodyS15}, where it found applications in lower bounds for set disjointness and circuit complexity, respectively.

\subsection{Our Results}
\label{sec:results}
Our work focuses on showing  $\widetilde{\Omega}(k)$ lower bounds on the randomized communication complexity of $k$-step pointer chasing
for protocols with an unlimited number of rounds,
which serves to show that the trivial $k$-round $O(k \log n)$-bit protocol described above is essentially tight.
Of course,  lower bounds with no restriction on rounds immediately imply lower bounds for any protocol with restricted rounds.

For protocols with constant, non-zero error, we show:
\begin{theorem}[Informal]
    Any randomized protocol that solves $k$-step pointer chasing 
    with constant (non-zero) error
    must send $\Omega(k \log (n/k))$
    bits in expectation,
    even for the weaker version of the problem where the goal is merely to compute some (non-constant) function $f(v)$ where $v$ is the end of a $k$-walk starting from some fixed vertex.
    \label{thm:const_error}
\end{theorem}

This lower bound is nearly tight, given the trivial $k$-round $O(k \log n)$-bit protocol. 
Moreover, the theorem shows that when $k \geq \sqrt{n}$, restricting the number of rounds to be less than $k$ does not matter very much: even if we do not limit the number of rounds, a protocol will have to expend $\Omega(k \log (n/k))$ bits of communication to solve $k$-step pointer chasing, and when $k \geq \sqrt{n}$ this very nearly matches the $(k-1)$-round protocol from~\cite{NisanW91}, which sends $O( (n/k)\log n + k \log n) = O(k \log n)$ bits.
This contrasts sharply with the case $k \ll \sqrt{n}$, where prior work has shown that  $(k-1)$-round protocols have significantly higher communication complexity compared to $k$-round protocols.

Theorem~\ref{thm:const_error} trivially implies a lower bound for protocols that are limited to at most $k-1$ communication rounds. In addition, as noted in Section~\ref{sec:related}, the lower bounds from~\cite{Klauck00,MaoYZ24} can be adapted to apply to protocols with at most $k - 1$ rounds, yielding a lower bound of $\Omega(n/k)$. 
We can thus combine these two lower bounds to obtain the following:
\begin{corollary}[Informal]
For any balanced function $f$ (e.g., the parity function),
any randomized constant-error protocol for $k$-step pointer chasing that uses at most $k - 1$ rounds must send $\Omega( (n/k) + k\log(n/k) )$ bits in expectation.
\end{corollary}
\noindent This nearly matches the $O\left( \left( n / k + k \right) \log n \right)$ protocol of \cite{NisanW91} across all regimes of $k$.

\medskip 

For \emph{zero-error} randomized protocols,
where the output must always be correct,
we prove a slightly stronger lower bound:
\begin{theorem}
   Any zero-error randomized protocol that solves $k$-step pointer chasing 
    must send $\Omega(k \log \log k)$
    bits in expectation,
    even for the weaker version of the problem where the goal is merely to compute some (non-constant) function $f(v)$ where $v$ is the end of a $k$-walk starting from some fixed vertex.
    \label{thm:zero_error}
\end{theorem}

\section{Preliminaries}
\label{sec:prelim}
Throughout this paper, we denote $[n] = \set{1,\ldots,n}$.

\subsection{Two-Party Communication Complexity}
In two-party communication complexity, there are two players, Alice and Bob, each with a private input $X,Y$ (resp.).
The goal is for one of the two parties, specified in advance, to output the value of some function $F(X,Y)$ of their inputs, using as little communication as possible.
To that end, the parties engage in a \emph{communication protocol},
where they communicate back-and-forth for some number of rounds,
until eventually the party responsible for producing an output does so.
\footnote{It is common to require that \emph{both} parties learn the value of $F(X,Y)$, and this is without loss of generality (up to one bit of communication) when $F$ is a Boolean function and the number of rounds is unrestricted: the party that learns $F(X,Y)$ can simply send it to the other party. However, when the number of rounds is restricted, requiring both parties to learn $F(X,Y)$ can significantly increase the communication required.}

A communication protocol may be deterministic, or it may be randomized; in this paper we consider \emph{randomized public-coin} protocols, 
where both players have access to a shared uniformly-random string of arbitrarily large length. The \emph{worst-case communication cost}
of the protocol is the total number of bits exchanged between the players, in the worst case over inputs and random strings. The \emph{expected communication cost} of the protocol is the worst-case over inputs $X,Y$ of the expected number of total number of bits exchanged between the players (where the expectation is taken over the randomness).

A protocol is said to compute $F$ with  error $0 \le \epsilon \le 1$ if on any input $X,Y$,
the probability that the protocol produces the correct output $F(X,Y)$
is at least $1 - \epsilon$ (where the probability is taken over the randomness).
In the special case where $\epsilon = 0$, we refer to this protocol as a \emph{randomized zero-error} protocol for $F$.
The \emph{randomized $\epsilon$-error communication complexity} of $F$
is the minimum expected communication cost of 
any protocol that computes $F$ with error $\epsilon$. 

The \emph{round complexity} of a protocol is  the worst-case number of messages exchanged between the two parties, speaking in alternating order. Either Alice or Bob may speak first (this must be specified by the protocol).

\subsection{Pointer Chasing}

The pointer chasing problem is defined by applying two functions, each given as input to one of Alice or Bob,  in alternating order and for some fixed number of steps $k$.

\noindent This alternation is formally captured by the following definition:

\begin{definition} [Walk functions] \label{def:walks}
    Let $n \in \N$, and $N_A,N_B \in [n]^n$.
    For any $r \in \N$,
    we define functions $\Walk{A}{r}$, $\Walk{B}{r} \in [n]^n$ with respect to $N_A, N_B$
    via mutual recursion.
    For any $i \in [n]$:
    \begin{equation*}
        \Walk{A}{r}(i)
        =
        \begin{cases}
            i                      & r = 0 \\
            \Walk{B}{r-1}(N_A(i))  & r > 0,
        \end{cases}
    \end{equation*}
    and symmetrically:
    \begin{equation*}
        \Walk{B}{r}(i)
        =
        \begin{cases}
            i                      & r = 0 \\
            \Walk{A}{r-1}(N_B(i))  & r > 0.
        \end{cases}
    \end{equation*}
\end{definition}
In order to solve the pointer chasing problem, the parties need to output the value $\Walk{A}{k}(1)$ reached after $k$ steps starting from value $1$, or more generally, compute a predetermined function $f$ of this value:
\begin{definition} [Pointer chasing problem] \label{def:PC}
    Let $n, k \in \N$,
    and let $f : [n] \rightarrow \set{0,1}$.
    In the \emph{$k$-step pointer chasing problem} $\PCF{f}{n}{k}$,
    Alice and Bob are given functions $N_A, N_B \in [n]^n$, respectively,
    and both parties know the function $f$.
    Their goal is to compute 
    \begin{equation*}
        \PCF{f}{n}{k}(N_A,N_B)
        \;=\;
        f( \Walk{A}{k}(1)). 
    \end{equation*}
\end{definition}

Throughout the paper, we consider a class of functions we refer to as \emph{non-trivial}, which is a slightly stronger condition than non-constant.

\begin{definition}
    $f:[n]\rightarrow\{0,1\}$ is \emph{non-trivial} if $f$ is non-constant on $[n]\setminus \set{1}$, or in other words, if there exist $i,j \in [n] \setminus \{1\}$ such that $f(i) \neq f(j)$. 
\end{definition}

\paragraph{Graph-theoretic formulation.}
For our purposes it is convenient to use an equivalent, graph-theoretic 
definition of the pointer chasing problem.
Let $G = (L, R, E)$ be a balanced directed bipartite graph where $|L| = |R| = \set{1,\ldots,n}$, and each vertex has out-degree 1.
If we think of $N_A$ (and resp.\ $N_B$) as the
edges going from $L$ to $R$
(resp.\ from $R$ to $L$),
then it is easy to see that $\Walk{A}{r}(1)$ (resp. $\Walk{B}{r}(1)$) returns the label of the vertex at the end of the unique $r$-step \textit{walk} in $G$ starting from the vertex labeled $1$ on the left side $L$ (resp.\ the right side $R$).
Thus, we can think of $\PCF{f}{n}{k}$
as the problem where Alice is given the edges $N_A$ from left to right, Bob is given the edges $N_B$ from right to left, and the goal is to compute the function $f$ of the vertex reached by a length-$k$ walk starting from vertex $1$ on the left side.

\subsection{The Problems $\Index$ and $\ORIndex$}

We use in our proofs   a reduction from the $\ORIndex$ function defined below.

\begin{definition}[The $\Index$ problem \cite{NK97}] 
Let $x \in [r]$ and $y \in \binset^r$.
The $\Index_r$ function is defined as:
\begin{equation*}
    \Index_{r}(x,y)
    =
    y_x.
\end{equation*}
\end{definition}
The $\ORIndex$ problem is simply the disjunction of $m$ instances of the Index function:
\begin{definition}[The $\ORIndex$ problem]
Let $x \in [r]^m$, and $y \in \left(\binset^r\right)^m$.
The $\ORIndex_{r,m}$ function is defined as
\begin{equation*}
    \ORIndex_{r,m}(x,y)
    =
    \bigvee_{i=1}^m 
    \left( 
        y_i
    \right)_{x_i}.
\end{equation*}
\end{definition}

In~\cite{Patrascu11} the following lower bound for $\ORIndex$ (which is called \emph{blocky lopsided disjointness} in~\cite{Patrascu11}) is given:
\begin{lemma} [\cite{Patrascu11} 
\label{lem:lower_or_index}
Theorem 1.4, restated and simplified]
    The randomized public-coin communication complexity
    of $\ORIndex_{r,m}$ with error $\frac{1}{9999}$ 
    is $\Omega(m \log r)$.
\end{lemma}

\subsection{The $\Cycle{n}$ problem}

In \cite{RS95}, Raz and Spieker defined the following problem, which we refer to as the $\Cycle{n}$ problem:
\begin{definition}
Fix an integer $n \in \N$ and two vertex sets $U,W$ of size $|U| = |W| = n$.
In the $\Cycle{n}$ problem,
Alice and Bob are given perfect matchings $P_A, P_B$ (resp.) in the complete bipartite graph on $U \cup W$.
The goal is for the players to determine whether $P_A \cup P_B$ is a Hamiltonian cycle on $U \cup W$.
\end{definition}

Note that in the (multi)-graph $G = (U \cup W,P_A \cup P_B)$, each vertex has degree exactly $2$, and therefore $G$ is a non-empty collection of cycles (possibly including degenerate cycles of length $2$).
This characterization will be useful to us later.

Raz and Spieker showed in~\cite{RS95} (Theorem 1) that the nondeterministic communication complexity
\footnote{In \emph{nondeterministic communication complexity}, there is a prover whose goal is to convince the two parties to output 1. We do not give a formal definition here, as it is not needed for our purposes. See \cite{NK97,RY20} for the formal definition.}
of the $\Cycle{n}$ problem is $\Omega(n\log\log{n})$.
Nondeterministic communication complexity is a \emph{lower bound} on zero-error randomized communication complexity (see \cite{NK97} Proposition 3.7), and the following statement immediately follows:

\begin{lemma}[Corollary of \cite{RS95} Theorem 1]
\label{cor:cycle_complexity} The randomized zero-error communication complexity of $\Cycle{n}$ is $\Omega(n \log\log{n})$.
\end{lemma}

\section{Lower Bounds}

\subsection{Lower bound for Constant Error Randomized Protocols}
\label{sec:sublinear_k}
In this section we prove our main lower bound on the communication complexity of $k$-step pointer chasing with unlimited interaction:

\begin{theorem} [Formal statement of Theorem~\ref{thm:const_error}]
    Let $k,n \in \N$ be such that $k \leq n$,
    and let $f:[n] \rightarrow \binset$ be non-trivial.
    Then
    there is a constant $\epsilon \in (0,1)$
    such that
    the randomized $\epsilon$-error
    communication complexity of
    $\PCF{f}{n}{k}$ 
    is
    $\Omega(k \log(n/k))$.
    \label{thm:LB_sublinear_k}
\end{theorem}

If we take $k = O(n^{\delta})$ for some constant $\delta \in (0,1)$,
then the bound that we obtain from Theorem~\ref{thm:LB_sublinear_k}
is $\Omega(k \log n)$,
which matches the na\"ive protocol described in Section~\ref{sec:intro}.

\subsubsection{Proof Overview.}
We show a reduction from  $\ORIndex_{r,m}$ to $\PCF{f}{n}{k}$, where we take $r,m$ such that $n = \Theta(rm)$ and $k = \Theta(m)$. 
Since $\ORIndex_{r,m}$ requires $\Omega(m \log r)$ bits of communication~\cite{Patrascu11},
this yields a lower bound of $\Omega(k \log(n/k))$ for $\PCF{f}{n}{k}$.

Recall that $\ORIndex_{r,m}$ comprises $m$ instances of $\Index_r$,
and the goal is to determine whether at least one of them evaluates to 1.
For convenience, denote these instances $\Index^{(1)}_r,\dots,\Index^{(m)}_r$, and the value of the $i$th instance as $\Index^{(i)}_r(X,Y)$.
We construct  edge sets $N_A, N_B$ such that $\PCF{f}{n}{k}(N_A, N_B) = 1$
iff there is at least one $i$ such that $\Index^{(i)}_r(X,Y) = 1$.

We describe the construction in graph-theoretic terms (see the explanation in Section~\ref{sec:prelim} above).
The graph that we construct is made up of $r$ gadgets, each associated with one instance $\Index^{(i)}_r$.
In addition, the graph contains two special vertices $\disj,\intersect$, which are not part of any gadget.
The gadgets are constructed such that
for each $i = 1,\ldots,m$,
if $\Index^{(i)}_r(X,Y) = 1$ then
a two-step walk on the $i$th gadget leads
to vertex $\intersect$;
but 
if $\Index^{(i)}_r(X,Y) = 0$,
then the two-step walk on the $i$th gadget leads to the next gadget if $i < m$,
or to vertex $\disj$ if $i = m$ (i.e., this is the last gadget).
Consequently, we can show that
$\ORIndex_{r,m}(X,Y) = 0$ if and only if $\Walk{A}{k}(1) = \disj$, and that $\ORIndex_{r,m}(X,Y) = 1$ if and only if $\Walk{A}{k}(1) = \intersect$.

\subsubsection{The Reduction.}

Fix parameters $m \geq 1, r \geq 2$.
We construct a protocol for
$\ORIndex_{r,m}$
by reduction to $\PCF{f}{n}{k}$,
where $f$ is any non-trivial function,
and $n,k \in \N$ satisfy
\begin{enumerate}
    \item $rm+2 \leq n$, and
    \item $2m \leq k$.
\end{enumerate}

Recall that in the $\ORIndex$ function,
Alice receives $X \in [r]^m$ and Bob receives $Y \in \left( \binset^r \right)^m$.
We construct a pointer-chasing instance as follows.
Let $\disj$, $\intersect \in [n]\setminus \set{1}$ be two vertices such that $f(\disj) \neq f(\intersect)$.
(We know that such vertices exist, because $f$ is non-trivial, i.e., non-constant on $[n]\setminus \set{1}$.)
We assume w.l.o.g.\ 
that $\disj = rm+1$ and $\intersect = rm+2$,
as otherwise we can re-order the vertices.

The edges $N_A$ going from left to right in the pointer chasing instance are as follows.
\begin{itemize}
    \item For any $i \in [m]$, we set $N_A(r(i-1) + 1)
    =
     r(i-1) + x_i$.
     \item  $N_A(\disj) = \disj$, and $N_A(\intersect) = \intersect$.
     \item For any vertex $s \in [n]$ such that $N_A(s)$ is not defined by the other cases, we set $N_A(s) = 1$ (this is an arbitrary choice, and any other value would work just as well).
\end{itemize}

The edges $N_B$ going from right to left in the pointer chasing instance are as follows.
\begin{itemize}
    \item For any $i \in [m]$ and $j \in [r]$:
\begin{equation*}
    N_B(r(i-1) + j)
    =
    \begin{cases}
        \intersect & \text{ if } \left( y_i \right)_j = 1 \\
        ri+1       & \text{ if } \left( y_i \right)_j = 0,
    \end{cases}
\end{equation*}
    \item  $N_B(\disj) = \disj$, and $N_B(\intersect) = \intersect$.
    \item For any vertex $s \in [n]$ such that $N_B(s)$ is not defined by the other cases, we set $N_B(s) = 1$ (again, an arbitrary value).
\end{itemize}

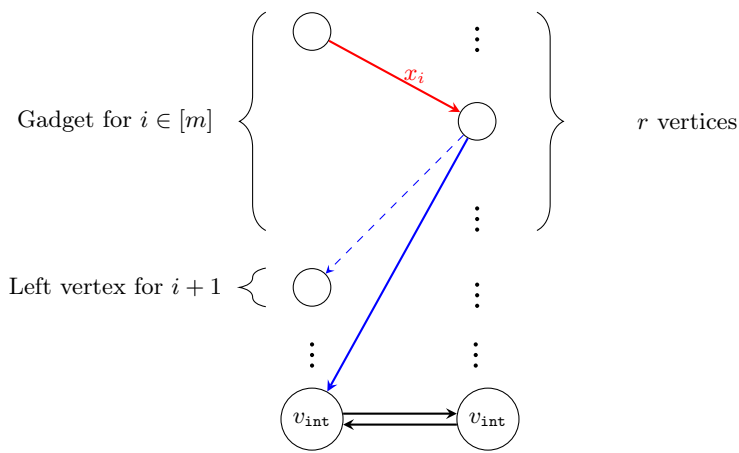
\begin{figure}[h]
\begin{tikzpicture}[
	vertex/.style={circle, draw=black, minimum size=5mm},
	hiddenNode/.style={circle, minimum size=5mm},
	arrow/.style={thick,->,>=stealth},
	weakArrow/.style={dashed,->,>=stealth},
	node distance=5mm and 15mm,
        ]
\node[vertex] (AliceI)  {};
\node[hiddenNode] (dotsI)       [right=of AliceI] {\huge $\vdots$};
\node[vertex] (BobI1)       [below=of dotsI] {};
\node[hiddenNode] (dotsI2)       [below=of BobI1] {\huge $\vdots$};
\node[hiddenNode] (HiddenNode) [left=of dotsI2] {};
\node[hiddenNode] (upperHiddenNodeI)     [left=1mm of AliceI] {};
\node[hiddenNode] (lowerHiddenNodeI)     [left=1mm of HiddenNode] {};
\node[hiddenNode] (upperRightHiddenNodeI)     [right=1mm of dotsI] {};
\node[hiddenNode] (lowerRightHiddenNodeI)     [right=1mm of dotsI2] {};

\node[vertex] (AliceIPlus1) [below=of HiddenNode] {};
\node[hiddenNode] (upperHiddenNodePlus1)     [left=1mm of AliceIPlus1] {};
\node[hiddenNode] (DotsPlus1)  [right=of AliceIPlus1] {\huge $\vdots$};

\node[hiddenNode] (lowerDots1)     [below=1mm of AliceIPlus1] {\huge $\vdots$};
\node[hiddenNode] (lowerDots2)     [right=13mm of lowerDots1] {\huge $\vdots$};

\node[vertex] (IntersectA)  [below=1mm of lowerDots1] {$\intersect$};
\node[vertex] (IntersectB)  [right=of IntersectA] {$\intersect$};

\draw[arrow, red] (AliceI) -- node[anchor=west] {$x_i$} (BobI1) ;
\draw[weakArrow, blue] (BobI1) -- node[anchor=east] {} (AliceIPlus1) ;
\draw[arrow, blue] (BobI1) -- node[anchor=west] {} (IntersectA) ;
\draw[arrow] (IntersectA.10) -- (IntersectB.170) ;
\draw[arrow] (IntersectB.190) -- (IntersectA.350) ;

\draw [decorate,decoration={brace,amplitude=10pt}]  (lowerHiddenNodeI.south) -- (upperHiddenNodeI.north)
node [black,midway,xshift=-2cm] {Gadget for $i \in [m]$};

\draw [decorate,decoration={brace,amplitude=10pt}]  (upperRightHiddenNodeI.north) -- (lowerRightHiddenNodeI.south)
node [black,midway,xshift=2cm] {$r$ vertices};

\draw [decorate,decoration={brace,amplitude=10pt}]  (upperHiddenNodePlus1.south) -- (upperHiddenNodePlus1.north)
node [black,midway,xshift=-2cm] {Left vertex for $i+1$};

\end{tikzpicture}
\caption{The left vertex and $x_i$-th right vertex of the gadget for coordinate $i \in [m]$, together with the left vertex of the gadget for coordinate $i+1$, and $\intersect$ vertices.
Each gadget has one left vertex and $r$ right vertices.
The left vertex of the gadget for coordinate $i$ is connected to the $x_i$-the right vertex.
For any $j \in [r]$, the $j$-th right vertex of the gadget for coordinate $i$ is connected to the $\intersect$ vertex if $\left(y_i\right)_j = 1$, and otherwise to the left vertex of the gadget for coordinate $i+1$.  
The edges between the $\intersect$ vertices always appear.}
\end{figure}

Observe that Alice can construct $N_A$ and Bob can construct $N_B$ from their respective inputs without communication.
Alice and Bob now solve the pointer-chasing instance $(N_A, N_B)$
by calling a protocol $\Pi$
for $\PCF{f}{n}{k}$,
and return $\ORIndex(x,y) = 0$ iff $\PCF{f}{n}{k}(N_A,N_B) = f(\disj)$.

\subsubsection{Analysis.}
Let $\Pi'$ be the protocol constructed above for $\ORIndex_{r,m}$.
Clearly, the communication complexity of $\Pi'$
is the same as that of the protocol $\Pi$ for $\PCF{f}{n}{k}$ (possibly plus one bit, depending on the party we want to output the function's value). We show that $\Pi'$ succeeds whenever $\Pi$ 
succeeds.

\begin{lemma} \label{lemma:sublinear_k_2_step_walk}
    Let $i \in [m]$.
    Then a two-step walk starting from vertex $r(i-1) + 1$ at the left side will end at vertex $ri+1$ if $\left(y_i\right)_{x_i}= 0$, and otherwise will end at the vertex $\intersect$.
    More formally, for any $i\geq 1$,
    \begin{equation*}
        \Walk{A}{2}(r(i-1) + 1)
        =
        \begin{cases}
            ri+1       & \text{ if }\left(y_i\right)_{x_i}= 0\\
            \intersect & \text{ otherwise } .
        \end{cases}
    \end{equation*}
\end{lemma}
\begin{proof}
    Follows immediately from the definition of $N_A(r(i-1)+1)$ and $N_B(r(i-1)+j)$.
\end{proof}

By induction on the length of the walk,
we can immediately deduce from Lemma~\ref{lemma:sublinear_k_2_step_walk}:
\begin{corollary}
\label{cor:sublinear_k_2_many_step_disjoint_walk}
    Let $0 \leq \ell \leq m$ be such that $\left(y_i\right)_{x_i}= 0$ for all $1 \leq i \leq \ell$.
    Then $\Walk{A}{2\ell}(1) = r\ell+1$.
\end{corollary} 

The correctness of the reduction follows:
\begin{corollary} \label{cor:reduction_for_sublinear_k}   
    If $\ORIndex(x,y) = 0$, then $\Walk{A}{k}(1)=\disj$, and otherwise $\Walk{A}{k}(1)=\intersect$. 
\end{corollary}
\begin{proof}
    First, assume $\ORIndex(x,y)=0$.
    Then by Corollary~\ref{cor:sublinear_k_2_many_step_disjoint_walk} for $\ell = m$, we have that $\Walk{A}{m}(1)=rm+1=\disj$. Since $r \leq k$, 
  and $N_A(\disj) = N_B(\disj) = \disj$, the claim follows.

    Next, assume that $\ORIndex(x,y)=1$, and let $i \in [m]$ be the smallest index such that $\left(y_i\right)_{x_i} = 1$.
    Then by Corollary~\ref{cor:sublinear_k_2_many_step_disjoint_walk} for $\ell = i-1$, we have that $\Walk{A}{2(i-1)}(1)=ri+1$.
    Lemma~\ref{lemma:sublinear_k_2_step_walk} implies that $\Walk{A}{2(i-1)+2}(1) = \intersect$.
    We note that $2(i-1)+2 \leq m \leq k$.
    Since $N_A(\intersect) = N_B(\intersect) = \intersect$, we conclude that $\Walk{A}{k}(1) = \intersect$.
\end{proof}

Taking $r = \Theta(n/k)$ and $m = \Theta(n)$,
Theorem~\ref{thm:LB_sublinear_k} now follows from the correctness of the reduction (Corollary~\ref{cor:reduction_for_sublinear_k}),
together with the lower bound of Lemma~\ref{lem:lower_or_index}
on the communication complexity of $\ORIndex_{r,m}$.

\subsection{Lower Bound for Zero-Error Randomized Protocols}
\label{sec:zero_error}
In this section we prove that any zero-error
protocol solving the pointer chasing problem requires $\Omega(k \log\log{k})$ bits of communication in expectation. 
\begin{theorem}[Formal statement of Theorem~\ref{thm:zero_error}]
    Let $n,k \in \N$ be such that $4k \leq n$, and
    let $f:[n] \rightarrow \{0,1\}$ be non-trivial. Then the randomized zero-error communication complexity of $\PCF{f}{n}{k}$
    is $\Omega(k \log\log{k})$.
    \label{thm:kllk}
\end{theorem}

It is convenient to assume that $f$ has the following properties.

\begin{observation}
\label{prop:zero_error_func}
Without loss of generality, we may assume that
    \begin{enumerate}[(a)]
        \item $|\{i \in [n] \mid f(i) = 0\}| \leq |\{i \in [n] \mid f(i) = 1\}|$,
        and 
        \item $f(1) = 0$, and $f(i) = 1$ for all $2 \leq i \leq 2k$.
    \end{enumerate}
\end{observation}
\begin{proof}
    Property (a) is obtained by relabeling the output values so that $1$ becomes the majority value.
    Given property (a), there are at least $2k$ indices such that $f(i) = 1$, and since $f$ is non-trivial, there exists an index $j \in [n] \setminus \{1\}$ such that $f(j) = 0$. Then, using a single extra step, it is easy to see we can set $f(1) = 0$ (by using the first step to move to vertex $j$), and $f(i) = 1$ for $2 \leq i \leq 2k$ (by reordering the indices).
\end{proof}

\subsubsection{Proof Overview.}
We show that for any integer $k$, we can reduce the $\Cycle{k}$ problem to the pointer-chasing problem $\PCF{f}{n}{2k'}$ for some $k < k' \leq 2k$ and for any $n \geq 4k$. We do so by first reducing $\Cycle{k}$ to $\Cycle{k'}$ for a prime $k < k' \leq 2k$, and then reducing $\Cycle{k'}$ to $\PCF{f}{n}{2k'}$.

In any input to $\Cycle{n}$,
each node lies on a single cycle.
Thus, the first reduction can be trivially achieved by 
adding $2(k' - k)$ fresh vertices
to the graph
to artificially extend the cycle on which  $u_1$ lies, so that its length is increased by 
 $2(k'-k)$.
 As a result, if there is a Hamiltonian cycle, its length now becomes $2k + 2(k' - k) = 2k'$ (recall that the original bipartite graph is over $2k$ vertices, $k$ on each side);
 and if there is no Hamiltonian cycle, extending the cycle on which $u_1$ lies by adding fresh vertices will not create one.

 The second reduction is done by constructing a pointer chasing instance that outputs $0$ if a $2k'$-walk starting from $u_1$ 
 returns to $u_1$, and outputs $1$ otherwise. We then prove that since $k'$ is prime, a $2k'$-walk from $u_1$ terminates at $u_1$ if and only if the cycle is Hamiltonian (excluding a case where $u_i$'s cycle is of length $2$, which can be handled trivially), and the theorem follows. As is clear from the formal definition of the reduction below, the reduction itself does not require the players to communicate.

\subsubsection{The Reduction.}
Let $U = \{u_1,\dots,u_k\}, W = \{w_1,\dots,w_k\}$. Alice and Bob are given perfect matchings $P_A,P_B$ between $U$ and $W$, respectively. For $i \in [k]$, we denote by $P_A(i)$  the index of the unique vertex $w_j$ neighboring $u_i$, i.e., $P_A(i) = j$ if $\{u_i,w_j\} \in P_A$. Similarly, for $w \in W$, we denote $P_B(j) = i$ if $\{u_i,w_j\} \in P_B$. Similar to the pointer chasing problem, we define a recursive notion of walk on the two matchings: for every $i \in [k]$ we denote

    \begin{equation*}
        \CycleWalk{A}{r}(i)
        =
        \begin{cases}
            i                      & r = 0 \\
            \CycleWalk{B}{r-1}(P_A(i))  & r > 0,
        \end{cases}
    \end{equation*}

    and symmetrically:
    \begin{equation*}
        \CycleWalk{B}{r}(i)
        =
        \begin{cases}
            i                      & r = 0 \\
            \CycleWalk{A}{r-1}(P_B(i))  & r > 0.
        \end{cases}
    \end{equation*}

First, we reduce $\Cycle{k}$ to $\Cycle{k'}$, where $k'$ is the closest prime from above to $k$, i.e., $k < k' \leq 2k$. We do so by artificially increasing the length of the cycle on which $u_1$ lies by $2(k'-k)-1$ edges. More formally: 
\begin{itemize}
    \item We add new vertices $x_1,\dots,x_{k'-k},y_1,\dots,y_{k'-k}$ to obtain new vertex sets, of size $k'$ each:
    \begin{align*}
        & U' = \{u_1,\dots,u_k\} \cup \{x_1,x_2,\dots,x_{k'-k}\},
        \\
        & W' = \{u_1,\dots,u_n\} \cup \{y_1,y_2,\dots,y_{k'-k}\}.
    \end{align*}
    \item Alice removes  the edge $\{u_1,w_{P_A(1)}\}$ from $P_A$.
    \item We add the path  $(u_1,y_1,x_1,y_2,x_2,\dots,x_{k'-k},w_{P_A(1)})$ to the graph, by adding the odd edges to $P_A$ and the even edges to $P_B$.
    
\end{itemize} 
The new graph is 
  a Hamiltonian cycle if and only if the original graph is a Hamiltonian cycle. Moreover, since we added $k'-k$ new vertices to each of the sets $U$, $W$
  to obtain $U', W'$,
  the result is indeed an instance of $\Cycle{k'}$. This concludes the reduction from $\Cycle{k}$
  to $\Cycle{k'}$; observe that the reduction can be performed   by Alice and Bob independently, with no communication.

Next, we reduce from $\Cycle{k'}$ to $\PCF{f}{n}{2k'}$ for any $n \geq 2k'$. 
The players start by checking if the cycle that contains $u_1$ is of length $2$. In order to do so, Alice first sends $P_A(1)$, and Bob sends $P_B(P_A(1))$, using $2 \lceil \log k' \rceil$ of communication in total.
If $P_B(P_A(1)) = 1$, the players conclude that $P_A \cup P_B$ is not a Hamiltonian cycle.
Otherwise, the players can conclude that the cycle containing $u_1$ is of length more than $2$.

Next, for $1 \leq i \leq k'$, define $N_A(i) = P_A(i)$ and $N_B(i) = P_B(i)$. For any $i > k'$, define $N_A(i),N_B(i)$ arbitrarily. In the next two lemmas, we show that the value of this instance of $\PCF{f}{n}{2k'}$ is equal to the value of the $\Cycle{k'}$ problem, and hence to the original $\Cycle{k}$ problem as well.

\begin{lemma}
\label{lem:cycle_to_normal_walk}
    For any $r \in \N \cup \{0\}$ we have $\CycleWalk{A}{r}(1) = \Walk{A}{r}(1)$.
\end{lemma}
\begin{proof}
    By induction on $r$. For $r = 0$, $\CycleWalk{A}{0}(1) = \Walk{A}{0}(1) = 1$. Assume by induction on $r$ that $\CycleWalk{A}{r-1}(1) = \Walk{A}{r-1}(1)$. For even $r$, 
    \[\CycleWalk{A}{r}(1) = P_B(\CycleWalk{A}{r-1}(1)),\text{ and }\Walk{A}{r}(1) = N_B(\Walk{A}{r-1}(1)).\] Since $\CycleWalk{A}{r-1}(1) = \Walk{A}{r-1}(1)$, then in particular $\Walk{A}{r-1}(1) \in [k']$. By choice of $N_B$, $P_B(i) = N_B(i)$ for all $1 \leq i \leq k'$, and the claim follows. The case for odd $r$ follows using a similar argument.
\end{proof}

\begin{lemma}
    $P_A \cup P_B$ is a Hamiltonian cycle on $U' \cup W'$ if and only if $f(\Walk{A}{2k'}(1)) = 0$. 
\end{lemma}
\begin{proof}
    If $P_A \cup P_B$ is a Hamiltonian cycle (i.e., a cycle of length $2k'$), then a $2k'$-walk from $u_1$ terminates at $u_1$, or in other words $\CycleWalk{A}{2k'}(1) = 1$. By Lemma~\ref{lem:cycle_to_normal_walk}, it holds that $\Walk{A}{2k'}(1) = \CycleWalk{A}{2k'}(1) = 1$,
    and therefore by Property~\ref{prop:zero_error_func}
    we have $f(\Walk{A}{2k'}(1)) = 0$.

    If $P_A \cup P_B$ is not a Hamiltonian cycle, then the cycle containing $u_1$ is of length $2\ell$ for some $1 < \ell < k'$.
    Since $k'$ is a prime number,  $\ell$ does not divide $k'$, implying that a $2k'$-walk from $u_1$ does not terminate at $u_1$. In other words, $\CycleWalk{A}{2k'}(1) \in [k'] \setminus \{1\}$. By Lemma~\ref{lem:cycle_to_normal_walk}, it holds that $\Walk{A}{2k'}(1) = \CycleWalk{A}{2k'}(1) \in [k'] \setminus \{1\}$, and therefore by Property~\ref{prop:zero_error_func}, $f(\Walk{A}{2k'}(1)) = 1$.

\end{proof}

This completes the correctness of the reduction from $\Cycle{k'}$ to $\PCF{f}{n}{2k'}$. Theorem~\ref{thm:kllk} immediately follows, by Lemma~\ref{cor:cycle_complexity}.

\bibliographystyle{splncs04}
\bibliography{related_work_refs}

\appendix

\end{document}